\documentclass[runningheads]{llncs}
\pdfoutput=1
\usepackage{subfigure,enumerate,amsmath,amssymb,graphicx}

\newtheorem{dfn}{Definition}

\newtheorem{prb}{Problem}
\newtheorem{prop}{Proposition}
\newtheorem{my_example}{Example}

\begin{document}

\mainmatter

\title{Improved Algorithms for the Point-Set Embeddability problem for Plane $3$-Trees}

\author{Tanaeem M Moosa \and M. Sohel Rahman }

\institute{Department of CSE, BUET, Dhaka-1000, Bangladesh\\
\texttt{http://www.buet.ac.bd/cse}
\email{ \and \{tanaeem,msrahman\}@cse.buet.ac.bd }
}

\authorrunning{Moosa and Rahman}

\maketitle

\setcounter{footnote}{0}

\begin{abstract}
In the point set embeddability problem, we are given a plane graph $G$ with $n$ vertices and a point set $S$ with $n$ points. Now the goal is to answer the question whether there exists a straight-line drawing of $G$ such that each vertex is represented as a distinct point of $S$ as well as to provide an embedding if one does exist. Recently, in~\cite{DBLP:conf/gd/NishatMR10}, a complete characterization for this problem on a special class of graphs known as the plane 3-trees was presented along with an efficient algorithm to solve the problem. In this paper, we use the same characterization to devise an improved algorithm for the same problem. Much of the efficiency we achieve comes from clever uses of the triangular range search technique. We also study a generalized version of the problem and present improved algorithms for this version of the problem as well.
\end{abstract}

\section{Introduction}
A \emph{planar graph} is a graph that can be \emph{embedded} in the plane, i.e., it can be drawn on the plane in such a way that its edges intersect only at their endpoints. A planar graph already drawn in the plane without edge intersections is called a \emph{plane graph} or \emph{planar embedding} of the graph. A straight-line drawing $\Gamma$ of a plane graph $G$ is a graph embedding in which each vertex is drawn as a point and each edge is drawn as straight line segments (as opposed to curves, etc.).

Given a plane graph $G$ with $n$ vertices and a set $S$ of $n$ points in the plane, a point-set embedding of $G$ on $S$ is a straight-line drawing of $G$ such that
each vertex is represented as a distinct point of $S$. The problem of computing a point-set embedding of a graph, also referred to as the point-set embeddability problem in the literature, has been extensively studied both when
the mapping of the vertices to the points is chosen by the drawing algorithm and
when it is partially or completely given as part of the input. There exists a number of results of the point-set embeddability problem on different graph classes in the literature~\cite{DBLP:journals/comgeo/Bose02,DBLP:journals/dcg/IkebePTT94,DBLP:journals/dam/KanekoK00,DBLP:journals/gc/PachW01}. A number of variants of the original problem have also been studied in the literature. For example in~\cite{DBLP:conf/wads/BadentGL07,DBLP:journals/jgaa/GiacomoDLMTW08}, a variant of the point-set embeddability problem has been studied, where the vertex set of the given graph and the given set of points are divided into a number of partitions and a particular vertex subset is to be mapped to a particular point subset. Other variants have also been studied with great interest~\cite{DBLP:journals/ijcga/KanekoK05,DBLP:journals/ijfcs/GiacomoLT06}.

Very recently, Nishat et al.~\cite{DBLP:conf/gd/NishatMR10} studied the point set embeddability problem on a class of graphs known as the \emph{plane 3-tree}.
Plane 3-trees belong to an interesting class of graphs and recently a number of different drawing algorithms have been presented in the literature on plane 3-trees~\cite{BiedlV10,MondalNRA10,DBLP:conf/gd/NishatMR10}. In this paper, we follow up the work of \cite{DBLP:conf/gd/NishatMR10} and improve upon their result from an algorithmic point of view. In~\cite{DBLP:conf/gd/NishatMR10}, Nishat et al. presented an $O(n^2 \log n)$ time algorithm that can decide whether a plane 3-tree $G$ of $n$ vertices admits a point-set embedding on a given set of $n$ points or not and compute a point-set embedding of $G$ if such an embedding exists. In this paper, we show how to improve the running time of the above algorithm. In particular, we take their algorithmic ideas as the building block of our algorithm and with some non trivial modifications we achieve a running time of $O(n^{4/3+\epsilon}\log n)$.
The efficiency of our algorithm comes mainly from clever uses of triangular range search and counting queries \cite{Paterson1986441,ALGOR::ChazelleSW1992,journals/dcg/Chazelle97} and bounding the number of such queries. Furthermore, we study a generalized version of the Point-Set Embeddability problem where the point set $S$ has more points than the number of vertices of the input plane 3-tree, i.e., $|S| = k >n$. For this version of the problem, an $O(nk^8)$ time algorithm was presented in~\cite{DBLP:conf/gd/NishatMR10}. We present an improved algorithm running in $O(nk^4)$ time.

The rest of this paper is organized as follows. Section~\ref{sec:prel} presents some definitions
and preliminary results. Section~\ref{sec:quad} presents a brief review of the algorithm presented in~\cite{DBLP:conf/gd/NishatMR10}. In Section~\ref{sec:subq} we present our main result. Section~\ref{sec:genProb} briefly discusses about the generalized version of the problem and we briefly conclude in Section~\ref{sec:con}.

%


\section{Preliminaries}\label{sec:prel}
In this section we present some preliminary notations, definitions and results that we use in our paper. We follow mainly the definitions and notations of~\cite{nishizeki2004planar}. We start with a formal definition of the \emph{straight-line drawing}.
\begin{dfn} [Straight-Line Drawing] \label{prb:strln}
Given a plane graph $G$, a straight-line drawing $\Gamma(G)$ of $G$ is a drawing of $G$ where vertices are drawn as points and edges are drawn as connecting lines.
\end{dfn}
The problem we handle in this paper is formally defined as follows.
\begin{prb} [Point-Set Embeddability] \label{prb:ptemb}
 Let $G$ be a plane graph of $n$ vertices and $S$ be a set of $n$ points on plane. The point-set embeddability problem wants to find a straight-line drawing of $G$ such that the vertices of $G$ are mapped to the points of $S$.
\end{prb}
Finding a point-set embedding for an arbitrary plane graph is proved to be NP-Complete~\cite{journals/jgaa/Cabello06}, even for some restricted subclasses. On the other hand, polynomial time algorithm exists for finding point-set embedding for outerplanar graphs or trees~\cite{journals/dcg/IkebePTT94,DBLP:journals/comgeo/Bose02}. An interesting research direction in the literature is to investigate this problem on various other restricted graph classes. One such interesting graph class, known as the plane 3-tree, is formally defined below.

\begin{definition}[Plane $3$-Tree] \label{def:plane3t}
A plane 3-tree is a triangulated plane graph $G=(V,E)$ with $n$ vertices such that either $n=3$, or there exists a vertex $x$ such that the graph induced by $V-\{x\}$ is also a plane $3$-tree.
\end{definition}

\begin{figure}[t!]
\begin{center}
\leavevmode
\scalebox{0.7}{
\includegraphics{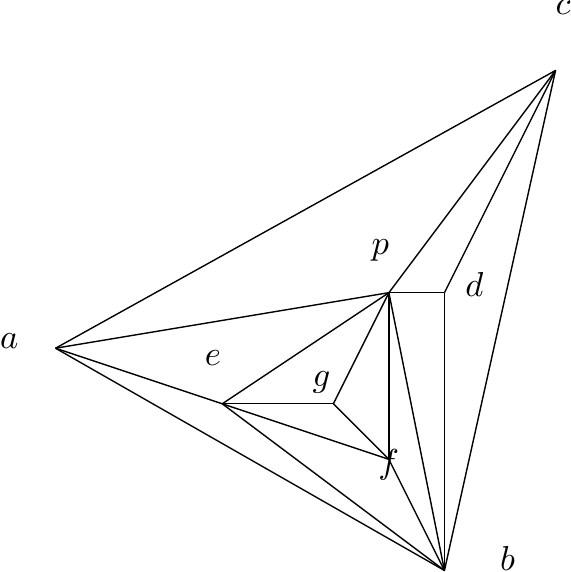}
}
\end{center}
\caption{A plane 3-tree of 17 vertices}
\label{fig:plane3tree}
\end{figure}

Figure~\ref{fig:plane3tree} presents a plane 3-tree with 17 vertices. As has been mentioned above, the very recent work of Nishat et al.~\cite{DBLP:conf/gd/NishatMR10} proved that finding point-set embedding is polynomially solvable if the input is restricted to a Plane $3$-Tree. Since a plane $3$-tree is triangulated, its outer face has only $3$ vertices, known as the \emph{outer vertices}. The following two interesting properties of a plane $3$-tree with $n>3$ will be required later.
\begin{prop} [\cite{conf/gd/BiedlV09}]
Let $G$ be a plane 3-tree with $n>3$ vertices. Then, there is a node $x$ with degree $3$ whose deletion will give a plane $3$-tree of $n-1$ vertices.
\end{prop}
\begin{prop} [\cite{conf/gd/BiedlV09}] \label{prop-P}
Let $G$ be a plane 3-tree with $n>3$ vertices. Then, there exists exactly $1$ vertex (say, $p$) that is a common neighbor of all $3$ outer vertices.
\end{prop}
%
For a plane 3-tree $G$, the vertex $p$ (as defined in Proposition~\ref{prop-P}) is referred to as the \textit{representative vertex} of $G$. For a plane graph $G$, and a cycle $C$ in it, we use $G(C)$ to denote the subgraph of $G$ inside $C$ (including $C$). In what follows, if a cycle $C$ is a triangle involving the vertices $x,y$ and $z$, we will often use $\triangle xyz$ and $G(\triangle xyz)$ to denote $C$ and $G(C)$. The following interesting lemma was recently proved in \cite{DBLP:conf/gd/NishatMR10} and will be useful later in this paper.
\begin{lemma} [\cite{DBLP:conf/gd/NishatMR10}]
Let $G$ be a plane $3$-tree of $n>3$ vertices and $C$ be any triangle of $G$. Then, the subgraph $G(C)$ is a plane $3$-tree.
\end{lemma}
We now define an interesting structure related to a plane 3-tree, known as the \emph{representative tree}.
\begin{definition} [Representative Tree]\label{defn-rtree}
Let $G$ be a plane $3$-tree with $n$ vertices with outer vertices $a$, $b$ and $c$ and representative vertex $p$ (if $n>3$). 
The representative tree $T$ of $G$ is an ordered tree defined as follows:
\begin{itemize}
\item If $n=3$, then $T$ is an single vertex.
\item  Otherwise, the root of $T$ is $p$ and its subtrees are the representative trees of $G(\triangle apb)$, $G(\triangle bpc)$ and $G(\triangle cpa)$ in that order.
\end{itemize}
\end{definition}

\begin{figure}[t!]
\begin{center}
\leavevmode
\scalebox{0.75}{
\includegraphics{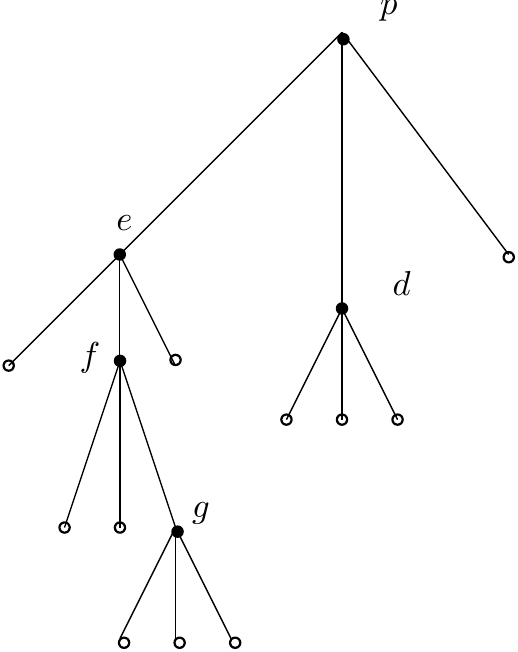}
}
\end{center}
\caption{The representative tree of the plane 3-tree of Figure~\ref{fig:plane3tree}}
\label{fig:rtree}
\end{figure}
The representative tree of the plane 3-tree of Figure~\ref{fig:plane3tree} is presented in Figure~\ref{fig:rtree}. Note that, the representative tree $T$ has $n'=n-3$ internal nodes, each internal node having degree $3$. Also, note that the outer vertices of $G$ are named as $a,b$ and $c$ respectively in counter-clockwise order around $p$. Therefore, the representative tree $T$ of a plane 3-tree $G$ is unique as per Definition~\ref{defn-rtree}. Now consider a plane 3-tree $G$ and its representative tree $T$. Assume that $G'$ is a subgraph of $G$ and $T'$ is a subtree of $T$. Then, $G'$ is referred to as the \emph{corresponding subgraph} of $T'$ if and only if $T'$ is the representative tree of $G'$. There is an $O(n)$ time algorithm to construct the representative tree from a given plane graph~\cite{DBLP:conf/gd/NishatMR10}.

Given a set of points $S$, we use the symbol $P_S(\triangle xyz)$ to denote the set of points that are inside the triangle $\triangle xyz$. We use the symbol  $N_S(\triangle xyz)$ to denote size of the set $P_S(\triangle xyz)$. We will extensively use the triangular range search and counting queries in our algorithm. Below we formally define these two types of queries.
\begin{prb}[Triangular Range Search] \label{prb:trs}
Given a set  $S$ of points that can be preprocessed, we have to answer queries of the form $SetQuery(S,\triangle abc)$, where the query returns $P_S(\triangle abc)$.
\end{prb}

\begin{prb}[Triangular Range Counting] \label{prb:trc}
Given a set  $S$ of points that can be preprocessed, we have to answer queries of the form  $CountQuery(S,\triangle abc)$, where the query returns $N_S(\triangle abc)$.
\end{prb}

In what follows, we will follow the following convention: If an algorithm has preprocessing time $f(n)$ and query time $g(n)$, we will say its overall running time is $\langle f(n), g(n) \rangle$. We conclude this section with an example illustrating the point set embedding of a plane 3-tree.
%
%
%

\begin{figure}[p!]
    \centering
    \subfigure[An example of the point set $S$.]
    {
    \scalebox{0.55}{
        \includegraphics{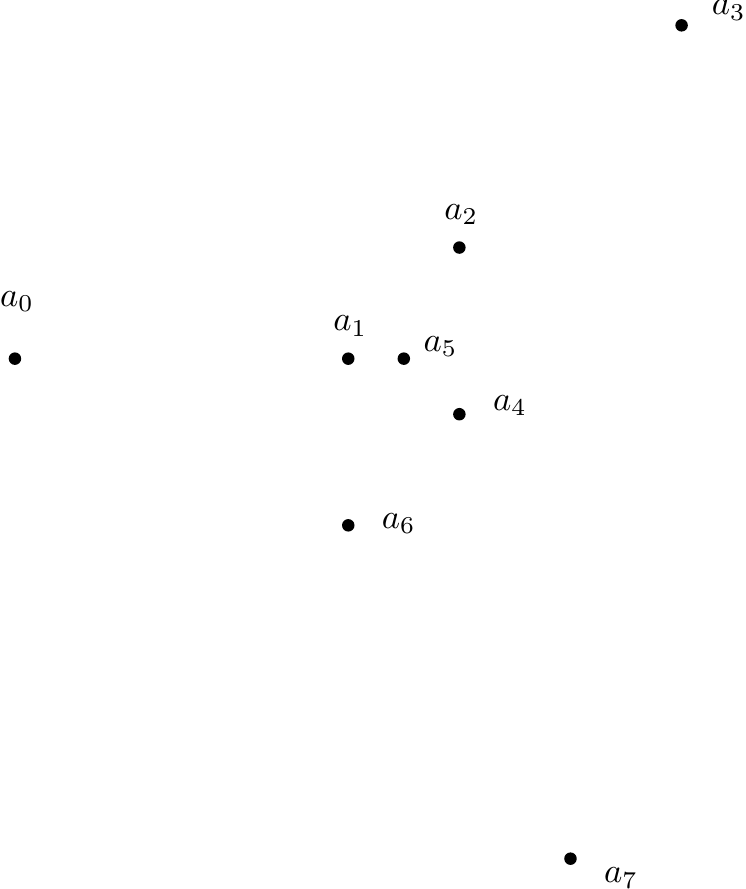}
    }
    \label{fig:pointset}
    }
    \\
    \subfigure[An embedding of the plane 3-tree of Figure~\ref{fig:plane3tree} on the point set of Figure~\ref{fig:pointset}]
    {
    \scalebox{0.6}{
        \includegraphics{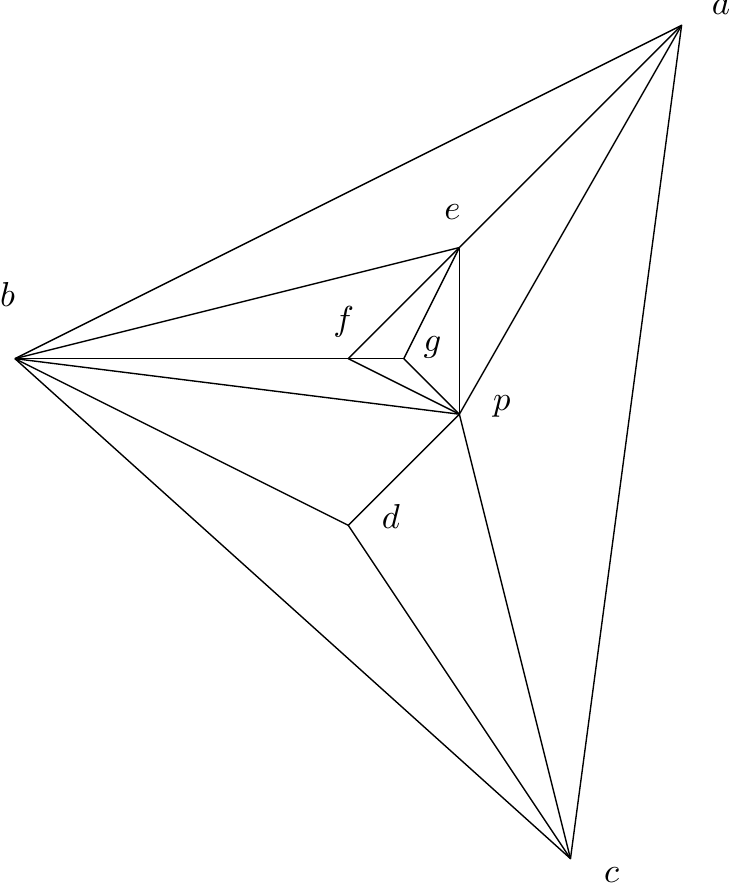}
    }
    \label{fig:embedding}
    }
    \caption{An example of point set embedding of a plane 3-tree.}
    \label{fig:example}
\end{figure}

\begin{my_example}
In Figure~\ref{fig:pointset} we present an example of the point set $S$ having $n = 17$ points. Then, in Figure~\ref{fig:embedding}, an embedding of the plane 3-tree of Figure~\ref{fig:plane3tree} is illustrated.
\end{my_example}

\section{Algorithm of~\cite{DBLP:conf/gd/NishatMR10}} \label{sec:quad}
In this section, we briefly describe the quadratic time algorithm of~\cite{DBLP:conf/gd/NishatMR10}. To simplify the description, we first assume that the points of $S$ are in general positions, i.e., no three points of $S$ are collinear.
\begin{lemma}[\cite{DBLP:conf/gd/NishatMR10}] \label{lmm:conv}
Let $G$ be a plane $3$-tree of $n$ vertices and $S$ be a set of $n$ points. If $G$ admits a point-set embedding on $S$, then the convex hull of $S$ contains exactly three points in $S$.
\end{lemma}

\begin{lemma}[\cite{DBLP:conf/gd/NishatMR10}] \label{lmm:uniqp}
Let $G$ be a plane $3$-tree of $n$ vertices with $a$, $b$ and $c$ being the three outer vertices of $G$, and let $p$ be the representative vertex of $G$. Let $S$ be a set of $n$ points such that the convex hull of S contains exactly three points. Assume that $G$ has a point-set embedding $\Gamma(G)$ on $S$ for a given mapping of $a$, $b$ and $c$ to the three points of the convex-hull of $S$. Then $p$ has a unique valid mapping.
\end{lemma}
The algorithm of~\cite{DBLP:conf/gd/NishatMR10} performs the following steps to find a valid embedding of $G$ in a given point-set $S$ if one exists.
\begin{enumerate}[Step: 1] \label{algo:quad}
\item Find the convex hull of the given points. If the convex hull does not have exactly $3$ points, return the message that no embedding exists.
\item For each of the possible $6$ mappings of the outer vertices of $G$ to the three points of the convex hull, perform Steps~\ref{step:start} and~\ref{step:start1} (recursively).
\item \label{step:start} Assume that at the beginning of this step, we are considering the representative (sub)tree $T'$ and the corresponding graph is $G'$ (obviously a subgraph of $G$). Let the three outer vertices of $G'$ are $a'$, $b'$ and $c'$ and the representative vertex of it is $p'$. Note that, initially, $G' = G$, $T' = T$ and the outer vertices and the representative vertex are $a$, $b$, $c$  and $p$ respectively. Assume that the number of internal nodes in $T'$ is $n'$. Note that, number of vertices in the corresponding graph $G'$ is $n'+3$. If $n'=0$ then embedding is trivially possible and this step returns immediately terminating the recursion. Otherwise, the following step is executed to check whether an embedding is indeed possible.
\item \label{step:start1} Let the root of $T'$ be $r$. Also let the three children of $r$ be $r_1$, $r_2$ and $r_3$ and the number of internal nodes in the subtrees rooted at $r_1$, $r_2$ and $r_3$ be $n_1'$, $n_2'$ and $n_3'$ respectively. Note that $n'=n_1'+n_2'+n_3'+1$. Let the three outer vertices $a'$, $b'$ and $c'$ of $G'$ are mapped to points $x$, $y$ and $z$ of $S$. Now, we find a point $u$ in $S$ such that $N_S(\triangle xuy)=n_1'$, $N_S(\triangle yuz)=n_2'$, and $N_S(\triangle zux)=n_3'$. By Lemma \ref{lmm:uniqp}, $u$ is unique if it exists. To find $u$, all the points of $S$ lying within the triangle $\triangle xyz$ are checked. If $u$ can be found, then $p$ is mapped to $u$ and Steps \ref{step:start} and \ref{step:start1} are executed recursively for all three subtrees of $T'$; otherwise no embedding is possible.
\end{enumerate}
In what follows, we will refer to this algorithm as the NMR Algorithm. Naive implementation of NMR algorithm runs in $O(n^3)$ time~\cite{DBLP:conf/gd/NishatMR10}. By sorting all points of $S$ according the polar angle with respect to each point of $S$ and employing some non-trivial observations, this algorithm can be made to run in $O(n^2)$ time~\cite{DBLP:conf/gd/NishatMR10}. Note that, the $O(n^2)$ algorithm assumes that the points of $S$ are at general positions. If this assumption is removed, NMR algorithm runs in $O(n^2\log n)$ time.

\section{Our Result} \label{sec:subq}
In this section, we modify the algorithm of~\cite{DBLP:conf/gd/NishatMR10} described in Section \ref{sec:quad} and achieve a better running time. 
For the ease of exposition, we assume for the time being that triangular range search has running time $\langle f(|S|),g(|S|)+\ell\rangle$ and triangular range counting has running time $\langle f(|S|),g(|S|)\rangle$, where $S$ is the input set and $\ell$ is the output size for triangular range search. We will finally use the actual running times during the analsysis of the algorithm. We first present our modified algorithm below followed by a detailed running time analysis.

\begin{enumerate}[Step 1:] \label{algo:subq}
\item \label{step:subq:chull}  Find the convex hull of the points of $S$. By Lemma \ref{lmm:conv} the convex hull should have $3$ points, otherwise no embedding exists.

\item \label{step:subq:trp} Preprocess the points of $S$ for triangular range search and triangular range counting.

\item \label{step:subq:map} For each of the possible $6$ mappings of the outer vertices of $G$ to the three points of the convex hull, perform Steps~\ref{step:subq:start} to~\ref{step:subq:check} (recursively).

\item \label{step:subq:start} We take the same assumptions as we took at Step~\ref{step:start} of the NMR algorithm.
Now, if $n'=0$ then embedding is trivially possible and this step returns immediately terminating the recursion. Otherwise, the following step is executed to check whether an embedding is indeed possible.


\item \label{step:subq:prune2}
Now we want to find a point $u$ such that $N_S(\triangle xuy)=n_1'$, $N_S(\triangle yuz)=n_2'$, and $N_S(\triangle zux)=n_3'$. Recall that, by lemma \ref{lmm:uniqp} such a point is unique if it exists.  Now, without loss of generality we can assume that $n_2' \leq \min(n_1', n_3')$. In order to find $u$, we first find points $v_1$ and $v_2$ on the line $yz$ such that $N_S(\triangle xv_1y)=n_1'$ and $N_S(\triangle xv_2z)=n_3'$. Note carefully that, in line $yz$, $v_1$ appears closer to $y$ than $v_2$; otherwise there will not be $n'$ points inside the triangle $\triangle xyz$. We will use a binary search and triangular range counting queries to find $v_1, v_2$ as follows. We first choose the mid point $w$ of the line $BC$. Then we compute $N_S(\triangle xwy)$ using a triangular range counting query. If $N_S(\triangle xwy) = n_1'$ we are done and we assign $v_1 = w$. Otherwise, if $N_S(\triangle xwy) > n_1'$ ($N_S(\triangle xwy) < n_1'$), then we choose the mid point $w'$ of the line $Bw$ ($wC$). Then we perform similar checks on $\triangle xw'y$. The point $v_2$ can also found similarly. Clearly, there always exist such points and steps of binary search is bounded by $O(\log N)$, where $N$ is the maximum absolute value of a point of $S$ in any coordinate.

\item \label{step:subq:check}
With points $v_1$ and $v_2$ at our disposal, we now try to find point $u$. Note that the point $u$ cannot be inside either $\triangle xv_1y$ or $\triangle xv_2z$. This is because if $u$ is in $\triangle xv_1y$ then $N_S(\triangle xuy) < N_S(\triangle xv_1y)= n_1'$ implying $N_S(\triangle xuy) < n_1'$, a contradiction. A similar argument is possible for $\triangle xv_2z$. So, we must have $u \in P_S(\triangle xv_1v_2)$. Also note that $N_S(\triangle xv_1v_2)=N_S(\triangle xyz)-N_S(\triangle xv_1y)-N_S(\triangle xv_2z)=n'-n_1'-n_3'=n_2'+1$. Using triangular range search we now find the points $P_S(\triangle xv_1v_2)$. To find $u$, we now simply check whether any of these points satisfies the requirement $N_S(\triangle xuy)=n_1'$, $N_S(\triangle yuz)=n_2'$, and $N_S(\triangle zux)=n_3'$. If no such point exists, then we return stating that it will be impossible to embed the graph on the points. Otherwise we find a point $u$, which is mapped to vertex $p$. Now Steps \ref{step:subq:start} to \ref{step:subq:check} are recursively executed for all three subtrees.

\end{enumerate}
\subsection{Analysis}\label{sec:analysis}
Now we analyze our modified algorithm presented above.
Step \ref{step:subq:chull} is executed once and can be done in  $O(n \log n)$ time. Step \ref{step:subq:trp} is executed once and can be done in  $f(|S|)$ time. Steps \ref{step:subq:start} to \ref{step:subq:check} are executed recursively. Step \ref{step:subq:start} basically gives us the terminating condition of the recursion. We focus on Step \ref{step:subq:prune2} and Step \ref{step:subq:check} separately below.

In Step \ref{step:subq:prune2}, we find out the two points $v_1$ and $v_2$ using binary search and triangular range counting queries. Time required for this step is $O(g(|S|)\log N)$. Note carefully that both the parameters $|S|$ and $N$ are constant in terms of recursion. Also, it is easy to see that, overall, Step \ref{step:subq:prune2} is executed once for each node in $T$. Hence, the overall running time of this step is $O(g(|S|)n\log N)$.

Now we focus on Step~\ref{step:subq:check}. Time required for triangular range search in Step \ref{step:subq:check} is  $O(g(|S|)+n_2')$. In this step we also need $O(n_2')$ triangular range counting queries which add up to $O(g(|S|)n_2')$ time. Recall that, $n_2' = \min(n_1', n_3')$, i.e., $n_2'$ is the number internal nodes of the subtree having the least number of internal nodes. Hence, we have $n_2' \leq n'/3$. Now, the overall running time of Step~\ref{step:subq:check} can be expressed using the following recursive formula: $T(n') = T(n_1')+T(n_2')+T(n_3')+ n_2' g(|S|))$, where $n_2' \leq \min(n_1',n_3')$. Now we have the following theorem:
        \begin{theorem}
        The overall running time of Step~\ref{step:subq:check} is $O(g(|S|)~ n\log n)$.
        \end{theorem}
        \begin{proof}
        First, assume that $T(n)\le c(n\log n) g(|S|), c\geq 1$. Then we have,
        \begin{align*}
            T(n') &= T(n_1')+T(n_2')+T(n_3')+ n_2' g(|S|)\\
	        &\leq c(n_1'\log n_1') g(|S|) + c(n_2'\log n_2') g(|S|) + c(n_3'\log n_3') g(|S|) +  n_2' g(|S|)\\
        	&< c(n_1'\log n') g(|S|) + c(n_2'\log \frac{n'}{2}) g(|S|) + c(n_3'\log n') g(|S|) +  c\times n_2' g(|S|)\\
        	&=  c(n_1'\log n') g(|S|) + c(n_2'(-1+\log n')) g(|S|) + c(n_3'\log n') g(|S|) +  cn_2' g(|S|) \\
        	&=c g(|S|) (-n_2'+(n_1' +n_2'+n_3') \log n' + n_2' )\\
        	&= c g(|S|) n' \log n'
        \end{align*}
        This completes the proof.
        \qed \end{proof}


Based on the above discussion, total time required for this algorithm is $O(n\log n + f(|S|)+ n g(|S|) \log N + n g(|S|) \log n) = O(f(|S|) + n~g(|S|) (\log n + \log N))$. Now we are ready to replace $f(|S|)$ and $g(|S|)$ with some concrete values. To the best of our knowledge, the best result of triangular range search and counting queries is due to Chazelle et al.~\cite{ALGOR::ChazelleSW1992}. In particular, Chazelle et al. proposed a solution for the triangular range search queries in~\cite{ALGOR::ChazelleSW1992} with time complexity $\langle O(m^{1+\epsilon}),O(n^{1+\epsilon}/m^{1/2})\rangle$, where $n<m<n^2$. Using this result the running time of our algorithm becomes $O(m^{1+\epsilon}+ (\log n + \log N) n^{2+\epsilon}/m^{1/2})$, which reduces to $O(n^{4/3+\epsilon}+ (\log n + \log N) n^{4/3+\epsilon})$ if we choose $m=n^{4/3}$.

Finally, we can safely ignore the $\log N$ component from our running time as follows. Firstly, the $\log N$ component becomes significant only when $N$ is doubly exponential in $n$ or larger, which is not really practical. Secondly, while we talk about the theoretical running time of algorithms, we often ignore the inherent $O(\log N)$ terms assuming that two (large) numbers can be compared in $O(1)$ time. For example, in comparison model, since sorting $n$ (large) numbers having maximum value $N$ requires $\Theta(n \log n)$ comparisons we usually say that sorting requires $\Theta(n \log n)$ time. Essentially, here, we ignore the fact that comparing two numbers actually requires $\Omega(\log N)$ time. Notably, the algorithm of ~\cite{DBLP:conf/gd/NishatMR10} also has an hidden $O\log N$ term since it requires $O(n^2)$ comparisons each of which actually requires $O(\log N)$ time. One final note is that for instances where the $\log N$ term does have significant effect, we can in fact get rid of the term using standard techniques to transform a counting algorithm into a ranking algorithm at the cost of a $\log n$ time increase in the running time. Similar techniques are also applicable for the algorithm of ~\cite{DBLP:conf/gd/NishatMR10}. So, we have the following theorem.

\begin{theorem}\label{thm-generalPos}
The point set Embeddability problem can be solved in $O(n^{4/3+\epsilon}\log n)$ time if the input graph is a plane 3-tree and $S$ does not contain any three points that are collinear.
\end{theorem}

%

\subsection{For points not in General positions}
So far we have assumed that the points of $S$ are in general positions, i.e., no three points in $S$ are collinear. We now discuss how to get around this assumption. Note that, the algorithm of Nishat et al~\cite{DBLP:conf/gd/NishatMR10} shows improved performance of $O(n^2)$ when the points of $S$ are in general positions. Now, if we remove our assumption, then we may have more than two points that are collinear. In this case, the only only modification needed in our algorithm is in Step \ref{step:subq:prune2}. Now, the problem is that the two points $v_1$ and $v_2$ could not be found readily as before. More specifically, even if Step \ref{step:subq:prune2} returns that $v_1$ and $v_2$ do not exist, still $u$ may exist. Now note that, in Step \ref{step:subq:prune2}, we want to find out $v_1$ and $v_2$ to ensure that $N_S(\triangle xv_1v_2)= n_2'+1$, where $n_2' = \min(n_1', n_3')$, i.e., $n_2' \leq n'/3$. Since, we have to check each points of $P_S(\triangle xv_1v_2)$ (to find $u$), the above bound of $n_2' \leq n'/3$ provides us with the required efficiency in our algorithm.

To achieve the same efficiency, we now slightly modify Step \ref{step:subq:prune2}. Suppose we are finding $v_1$ ($v_2$). We now try to find $v_1$ ($v_2$) such that the $N_S(\triangle xv_1y) > n_1'$ ($N_S(\triangle xv_2z) > n_3'$) and $v_1$ ($v_2$)  is as near as possible to $B$ ($C$) on the line $BC$. Let us assume that we need $\mathcal I$ iterations now to find $v_1$ ($v_2$). We have the following bound for $\mathcal I$.
\begin{lemma}
$\mathcal I$ is bounded by $O(\log N)$
\end{lemma}
\begin{proof}
There may not be any point candidate of $v_1$ ($v_2$) which has integer coordinates. But as $x$ can be intersection of two lines, each of which goes through two points of $S$, there may exists a point candidate of $v_1$ having denominator less than $N^2$ or there is none. Either way, to find such a point or to be sure no such point exists we only need precision less than $1/N^2$. Therefore, $O(\log N)$ iterations are sufficient.
\qed \end{proof}

Again, the argument presented at the end of Section~\ref{sec:analysis} about the component $\log N$ applies here as well. Therefor, the result of Theorem~\ref{thm-generalPos} holds even the points of $S$ are not in general positions. So, we restate our stronger and more general result as follows.

\begin{theorem}\label{thm-nogeneralPos}
The point set Embeddability problem can be solved in $O(n^{4/3+\epsilon}\log n)$ time if the input graph is a plane 3-tree.
\end{theorem}

\section{Generalized Version}\label{sec:genProb}
A generalized version of the Point Set Embeddability problem is also of interest in the graph drawing research community, where $S$ has more points than the number of vertices of the input graph $G$. More formally, the generalized point set embeddability problem is defined as follows.
\begin{prb} [Generalized Point-Set Embeddability] \label{prb:ptemb}
 Let $G$ be a plane graph of $n$ vertices and $S$ is a set of $k$ points on plane such that $k > n$. The generalized point-set embeddability problem wants to find a straight-line drawing of $G$ such that vertices of $G$ are mapped to some $n$ points of $S$.
\end{prb}
In this section, we extend our study to solve the Generalized Point-set Embeddability problem for plane 3-trees. This version of the problem was also handled in~\cite{DBLP:conf/gd/NishatMR10} for plane 3-trees and they presented an algorithm for the problem that runs in $O(nk^8)$ time. Our target again is to improve upon their algorithm.
We show how to improve the result to $O(n k^4)$.

We will use dynamic programming (DP) for this purpose. For DP, we define our (sub)problem to be $Embed(r',a',b',c')$ where $r'$ is the root of the (sub)tree $T'$ and $a'$,$b'$ and $c'$ are points in $S$. Now, $Embed(r',a',b',c')$ returns true if and only if it is possible to embed the corresponding subgraph $G'$ of the subtree $T'$ rooted at $r'$ such that its three outer vertices are mapped to the points $a'$, $b'$  and  $c'$. Now we can start building the DP matrix by computing $Embed(r',a',b',c')$ for the smaller subtrees to see whether the corresponding subgraphs can be embedded for a combination of 3 points of $S$ as outer vertices of the corresponding subgraphs. In the end, the goal is to check whether $Embed(r,a,b,c)$ returns true for any particular points $a,b,c \in S$, where $r$ is the root of the representative tree $T$ of the input plane 3-tree $G$. Clearly, if there exists $a,b,c\in S$ such that $Embed(r,a,b,c)$ is true, then $G$ is embeddable in $S$. We now have the following theorem.
\begin{theorem}
The Generalized Point-Set Embeddability problem can be solved in $O(n\times k^4)$ time.
\end{theorem}
\begin{proof}
If $r'$ is a leaf, then $Embed(r',a',b',c')$ is trivially true for any $a',b',c'\in S$. Now consider the calculation of $Embed(r',a',b',c')$ when $r'$ is not a leaf. Then, assume that the children of $r'$ are $r_1'$, $r_2'$ and $r_3'$ in that order. Then $Embed(r',a',b',c')$ is true if and only if we can find a point $u \in P_S(\triangle a'b'c')$ such that $Embed(r_1',a',b',u)$, $Embed(r_2',u,b',c')$ and $Embed(r_3',a',u,c')$ are true; otherwise $Embed(r',a',b',c')$ will be false. Clearly, the time required to calculate $Embed(r',a',b',c')$ is $O(g(|S|) + N_S(\triangle a'b'c'))$. Therefore, in the worst case the time required to calculate an entry of the DP matrix is $O(g(|S|) + k)$. Now, it is easy to realize that there are in total $nk^3$ entries in the DP matrix. Note that to be able to compute $P_S(\triangle a'b'c')$ we need to spend $O(f(|S|))$ time for a one-time preprocessing. Hence, the total running time is $O(f(|S|) + nk^3 \times (g(|S|) + k)) = O(f(|S|) + nk^3 g(|S|) + nk^4)$. Using the $\langle O(m^{1+\epsilon}),O(n^{1+\epsilon}/m^{1/2})\rangle$ result of Chazelle et al.~\cite{ALGOR::ChazelleSW1992}, the running time becomes $O(O(m^{1+\epsilon}) + nk^3 \times n^{1+\epsilon}/m^{1/2} + nk^4)$. Since, $n<m<n^2$, this running time is $O(nk^4)$ in the worst case.
\qed \end{proof}
\section{Conclusion}\label{sec:con}
In this paper, we have followed up the work of \cite{DBLP:conf/gd/NishatMR10} and presented an algorithm to solve the point-set embeddability problem in $O(n^{4/3+\epsilon} \log n)$ time.
This improves the recent $O(n^2 \log n)$ time result of~\cite{DBLP:conf/gd/NishatMR10}. Whether this algorithm can be improved further is an interesting open problem. In fact, 
an $o(n^{4/3})$ algorithm could be an interesting avenue to explore, which however, does not seem to be likely with our current technique. Since there are $\Omega(n)$ nodes in the tree, any solution that uses triangular range search to check validity at least once for each node in the tree would require $\Omega(n)$ calls to triangular range query. Lower bound for triangular range search is shown to be $\langle \Omega(m), \Omega(n/m^{1/2}) \rangle$~\cite{journals/dcg/Chazelle97}, which suggests an $\Omega(n^{4/3})$ lower bound for our algorithm using triangular range search.

We have also studied a generalized version of the point-set embeddability problem where the point set $S$ has more than $n$ points. For this version of the problem we have presented an algorithm that runs in $O(nk^4)$ time, where $k = |S|$. Nishat et al.~\cite{DBLP:conf/gd/NishatMR10} also handled this version of the problem and presented an $O(nk^8)$ time algorithm. It would be interesting to see whether further improvements in this case are possible. Also, future research may be directed towards solving these problems for various other classes of graphs.

\bibliographystyle{abbrv}
\bibliography{planetree}
\end{document}